\documentclass[
	paper=A4,
	pagesize=auto,
	fontsize=12pt
]{scrartcl}
\usepackage{libertine}
\recalctypearea

\usepackage[ngerman,british]{babel}
\usepackage[style=alphabetic,url=false,isbn=false,doi=true,eprint=false,backend=biber]{biblatex}
\usepackage{graphicx}
\usepackage{amsmath,amssymb,amsthm}
\usepackage{csquotes}
\usepackage{bm}

\usepackage{hyperref}
\hypersetup{
  colorlinks
}
\usepackage{cleveref}

\theoremstyle{plain}
\newtheorem{theorem}{Theorem}[section]
\newtheorem{lemma}[theorem]{Lemma}
\newtheorem{corollary}[theorem]{Corollary}

\theoremstyle{remark}
\newtheorem{example}[theorem]{Example}
\newtheorem{remark}[theorem]{Remark}

\theoremstyle{definition}
\newtheorem{definition}[theorem]{Definition}

\newlength{\JZHeightOfX}
\newcommand{\JZOrcidlink}[1]{
\setlength{\JZHeightOfX}{\fontcharht\font`X}
\includegraphics[height=\JZHeightOfX]{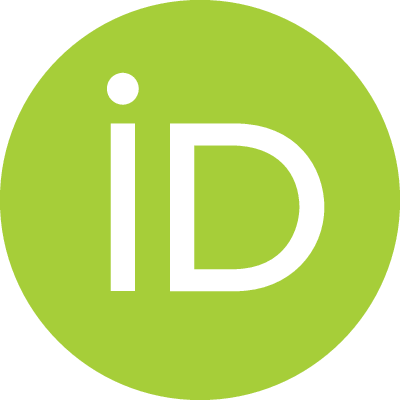}
\href{https://orcid.org/#1}{#1}
}

\begin{document}
\author{Jobst Ziebell\hspace{2em}\JZOrcidlink{0000-0002-9715-6356}\\
\small{Abbe Center of Photonics, Friedrich-Schiller-University, Jena, Germany}}
\title{$\theta$-splitting Densities and Reflection Positivity}
\date{\today}
\maketitle

\begin{abstract}
A simple condition is given that is sufficient to determine whether a measure that is absolutely continuous with respect to a Gaußian measure on the space of distributions is reflection positive.
It readily generalises conventional lattice results to an abstract setting, enabling the construction of many reflection positive measures that are not supported on lattices.
\end{abstract}
\section{Introduction}
Reflection positivity is one of the pillars of Euclidean quantum field theories.
It is readily established for wide sets of Gaußian measures but for non-Gaußian measures, the author feels that - with the exception of measures supported on lattices - there is no general framework that can be easily applied.
For measures that are absolutely continuous with respect to Gaußian measures, that is fixed in this article by introducing the set of $\theta$-splitting functions, which can work as densities to directly generalise the lattice methods used e.g. in \cite{src:GlimmJaffe}.
The result is very simple: Given a $\theta$-invariant reflection positive Gaußian measure and applying a measurable density to it that is $\theta$-splitting, the outcome is a reflection positive measure.

In general, physically relevant measures in $d \ge 3$ dimensions are typically not absolutely continuous with respect to the Gaußian free field measure.
Hence, one still needs to find ways to regularise the models of interest in order to apply the theorems in this work.
However, reflection positivity is preserved by the weak convergence of measures (see e.g. \cite{src:Bogachev:GaußianMeasures}), since it implies the pointwise convergence of corresponding characteristic functions.
Hence, any limit point of a sequence of regularised models corresponding to reflection positive measures is reflection positive as well.
\section{Preliminaries}
A \textbf{locally convex space} is a real topological vector space whose topology is induced by some family of seminorms.
The \textbf{dual} of a locally convex space $X$ equipped with the strong dual topology will be denoted by $X^*_\beta$.
\textbf{Inner products} denoted with round brackets $(\cdot,\cdot)$ are taken to be $\mathbb{R}$-bilinear.
Throughout this work, $d \in \mathbb{N}$ is fixed.
We shall work on the spaces
\begin{equation}
\mathcal{D} := \mathcal{D}(\mathbb{R}^{d+1})
\qquad\text{and}\qquad
\mathcal{D}_+ := \mathcal{D}(\mathbb{R}_{>0} \times \mathbb{R}^d) \, .
\end{equation}
of real test functions with their canonical LF topologies \cite[p. 131-133]{src:Treves:TopologicalVectorSpaces}.
Let us denote the corresponding continuous \textbf{restriction map} by $\pi_+ : \mathcal{D}^*_\beta \to (\mathcal{D}_+)^*_\beta$ (see e.g. \cite[p. 245-246]{src:Treves:TopologicalVectorSpaces}).
$\mathcal{D}$ and $\mathcal{D}_+$ as well as their strong duals $\mathcal{D}^*_\beta$ and $(\mathcal{D}_+)^*_\beta$ are complete \cite[Theorem 13.1]{src:Treves:TopologicalVectorSpaces}, barrelled \cite[p. 347]{src:Treves:TopologicalVectorSpaces}, nuclear spaces \cite[p. 530]{src:Treves:TopologicalVectorSpaces} (hence, reflexive by \cite[p. 147]{src:SchaeferWolff:TopologicalVectorSpaces}) that are also Lusin spaces \cite[p. 128]{src:Schwartz:RadonMeasures} and thus in particular Souslin spaces.
\begin{theorem}[{\cite[Lemma 6.4.2.(ii), Lemma 6.6.4]{src:Bogachev:MeasureTheory2}}]
\label{thm:SouslinBorelProduct}
Let $X$ and $Y$ be Souslin spaces.
Then, the Borel $\sigma$-algebra of $X \times Y$ coincides with the $\sigma$-algebra generated by all products of Borel sets in $X$ and $Y$ respectively.
\end{theorem}
In this work, a \textbf{measure} is taken to be a countably additive nonnegative function on a $\sigma$-algebra.
A \textbf{Borel measure} is thus a measure on a Borel $\sigma$-algebra and a \textbf{Radon measure} is a Borel measure that is inner regular over compact sets.
A \textbf{centred Gaußian measure} on a locally convex space $X$ is a Borel probability measure with the property that the pushforward measures by elements of $X^*$ are centred Gaußians or the Dirac delta measure $\delta_0$ at the origin.
One can in general consider non-Radon Gaußian measures on locally convex spaces.
However, every Borel measure on the spaces $\mathcal{D}^*_\beta, (\mathcal{D}_+)^*_\beta$ and countable products thereof is automatically Radon \cite[Theorem 7.4.3]{src:Bogachev:MeasureTheory2}.

A subset $A \subseteq X$ is \textbf{$\bm{\mu}$-measurable} with respect to a measure $\mu$ on some $\sigma$-algebra $\mathcal{A}$ on $X$ if it is in the Lebesgue completion $\mathcal{A}_\mu$ of $\mathcal{A}$ with respect to $\mu$.
Similarly, a function $f: X \to [-\infty, \infty]$ is \textbf{$\bm{\mu}$-measurable} if the preimage of every Borel subset of $[-\infty,\infty]$ is in $\mathcal{A}_\mu$.
Likewise, $f: X \to [-\infty, \infty]$ is \textbf{$\bm{\mu}$-integrable} if $f$ is $\mu$-measurable and $\int |f| \mathrm{d}\mu < \infty$.
A subset $A \subseteq X$ is \textbf{$\bm{\mu}$-negligible} if it is a subset of some $B \in \mathcal{A}$ with $\mu(B) = 0$.

The \textbf{pushforward} of a Borel measure $\mu$ on a Hausdorff space $X$ by a continuous function $f : X \to Y$ to a Hausdorff space $Y$ will be denoted by $f_* \mu$.
It is automatically a Borel measure on $Y$ and if $\mu$ is Radon, so is $f_* \mu$ \cite[Theorem 9.1.1.(i)]{src:Bogachev:MeasureTheory2}.
The \textbf{convolution} of two Borel measures $\mu$ and $\nu$ on a Souslin locally convex space $X$ is given by $\mu * \nu = s_*(\mu \times \nu)$ where $s : X \times X \to X, (x,y) \mapsto x+y$.
This is well-defined by \cref{thm:SouslinBorelProduct}.

To every finite Borel measure $\mu$ on a locally convex space $X$ we associate its \textbf{characteristic function} $\hat{\mu} : X^* \to \mathbb{C}$ with
\begin{equation}
\phi \mapsto \int_X \exp \left[ i \phi \left( x \right) \right] \mathrm{d} \mu \left( x \right) \, .
\end{equation}
It is well-known that two Radon measures on a locally convex space are equal if and only if their characteristic functions are equal \cite[Lemma 7.13.5]{src:Bogachev:MeasureTheory2}.
Moreover, if $\mu$ is a centred Gaußian measure on $X$, its characteristic function is given by
\begin{equation}
\hat{\mu} \left( \phi \right)
=
\exp \left[ - \frac{1}{2} \left( \phi, \phi \right)_{L^2(\mu)} \right]
\end{equation}
for all $\phi \in X^*$ \cite[Theorem 2.2.4, Corollary 2.2.5]{src:Bogachev:GaußianMeasures}.
\begin{theorem}
\label{thm:SouslinMeasurability}
Let $f : X \to Y$ be a continuous map from a Souslin space $X$ to a Hausdorff space $Y$.
Then, for every Borel set $B \subseteq X$, $f(B)$ is measurable by any Radon measure on $Y$.
\end{theorem}
\begin{proof}
Since every Borel subset of a Souslin space is Souslin \cite[p. 96 Theorem 3]{src:Schwartz:RadonMeasures}, this follows directly from \cite[Theorem A.3.15]{src:Bogachev:GaußianMeasures}.
\end{proof}
\begin{corollary}
\label{cor:PushforwardMeasurability}
Let $p : X \to Y$ be a continuous map from a Souslin space $X$ to a Hausdorff space $Y$ and $\mu$ a Radon measure on $X$.
Then every function $f : Y \to [-\infty,\infty]$ with the property that $f \circ p$ is $\mu$-measurable is $(p_*\mu)$-measurable.
\end{corollary}
\begin{proof}
First, note that $p(X)$ is $(p_*\mu)$-measurable by the preceeding theorem.
Now, letting $B \subset [-\infty,\infty]$ be a Borel set, we have
\begin{equation}
p^{-1} \left( f^{-1} \left( B \right) \right)
=
A \cup N_1
\end{equation}
for some Borel subset $A \subseteq X$ and some $\mu$-negligible set $N_1 \subseteq X$.
For brevity, let $N_2 = Y \setminus p(X)$, which is clearly $(p_* \mu)$-negligible.
Then,
\begin{equation}
\begin{aligned}
f^{-1} \left( B \right)
&=
\left[ f^{-1} \left( B \right) \cap p(X) \right]
\cup
\left[ f^{-1} \left( B \right) \cap N_2 \right] \\
&=
p \left( p^{-1} \left( f^{-1} \left( B \right) \right) \right)
\cup
\left[ f^{-1} \left( B \right) \cap N_2 \right] \\
&=
p \left(A \right) \cup p \left(N_1 \right) \cup \left[ f^{-1} \left( B \right) \cap N_2 \right] \, .
\end{aligned}
\end{equation}
$p (A )$ is $(p_* \mu)$-measurable by the preceeding theorem and $p (N_1 )$ as well as $f^{-1} ( B ) \cap N_2$ are clearly $(p_* \mu)$-negligible.
\end{proof}
We close this section by a simple lemma on positive semidefinite matrices.
\begin{lemma}[{\cite[Satz VII]{src:Schur:BilinearFormen}}]
\label{lem:PositiveSemidefinitenessPreserving}
Let $N \in \mathbb{N}$ and $A, B$ be positive semidefinite $N \times N$ matrices with respect to the standard inner product on $\mathbb{C}^N$.
Then the matrix $(A_{m,n} B_{m,n})_{m,n = 1}^N$ given by component-wise multiplication is positive semidefinite.
\end{lemma}
\begin{proof}
Diagonalising $B$ by a unitary matrix $U$, we obtain
\begin{equation}
B_{m,n}
=
\sum_{a = 1}^N U_{m,a}^* \lambda_a U_{n,a}
\end{equation}
for some nonegative numbers $\lambda_1, \dots, \lambda_N$.
Hence, for any $c \in \mathbb{C}^N$,
\begin{equation}
\sum_{m,n,a,b = 1}^N c_m^* A_{m,n} B_{m,n} c_n
=
\sum_{a = 1}^N \lambda_a \sum_{m,n = 1}^N \left( U_{m,a} c_m \right)^* A_{m,n} \left( U_{n,a} c_n \right)
\ge
0 \, .
\end{equation}
\end{proof}
\section{Reflection Positivity}
On $\mathbb{R}^{d+1}$ we define the operation of \textbf{time reflection} which we shall denote by $\theta: \mathbb{R}^{d+1} \to \mathbb{R}^{d+1}, (x_1, \dots, x_{d+1}) \mapsto (-x_1, x_2, \dots, x_{d+1})$.
By a slight abuse of notation, $\theta$ extends continuously and linearly to $\mathcal{D}$ and $\mathcal{D}^*_\beta$ in the obvious way.
\begin{definition}[{\cite[p. 90]{src:GlimmJaffe}}]
Let $\mu$ be a finite Borel measure on $\mathcal{D}^*_\beta$.
Then $\mu$ is \textbf{reflection positive} if for every sequence $(\phi_n)_{n \in \mathbb{N}}$ in $\mathcal{D}_+$, every sequence $(c_n)_{n \in \mathbb{N}}$ of complex numbers and every $N \in \mathbb{N}$,
\begin{equation}
\sum_{m,n = 1}^N c_m^* \hat{\mu} \left( \phi_m - \theta \phi_n \right) c_n \ge 0 \, .
\end{equation}
Furthermore, $\mu$ is \textbf{$\bm{\theta}$-invariant} if $\theta_* \mu = \mu$.
\end{definition}
To begin with, let us recapitulate two of the most important (in the author's opinion) theorems on reflection positive measures along with their proofs.
\begin{theorem}[{\cite[Theorem 6.2.3]{src:GlimmJaffe}}]
\label{thm:ReflectionPositiveL2}
Let $\mu$ be a finite, reflection positive Borel measure on $\mathcal{D}^*_\beta$ with the property that for every $\phi \in \mathcal{D}_+$ the function $\mathbb{R} \to \mathbb{C}, t \mapsto \hat{\mu}(t \phi)$ has an analytic continuation to some neighbourhood of zero in the complex plane.
Then, $( \phi, \theta \phi )_{L^2(\mu)} \ge 0$ for all $\phi \in \mathcal{D}_+$.
\end{theorem}
\begin{proof}
For $\lambda > 0$ let $\psi_1 = \lambda \phi$, $\psi_2 = 0$, $c_1 = \lambda^{-1}$ and $c_2 = -\lambda^{-1}$.
Since $\mu$ is reflection positive, we obtain
\begin{equation}
\begin{aligned}
0 &\le \sum_{m,n = 1}^{2} c_m^* \hat{\mu} \left( \psi_m - \theta \psi_n \right) c_n \\
&=
\frac{1}{\lambda^2} \int_{\mathcal{D}^*_\beta}
\left(
\exp \left[ i \lambda T \left( \phi - \theta \phi \right) \right]
-
\exp \left[ - i \lambda T \left( \phi \right) \right]
-
\exp \left[ - i \lambda T \left( \theta \phi \right) \right]
+
1
\right)
\mathrm{d} \mu \left( T \right) \, .
\end{aligned}
\end{equation}
By a classical theorem of Lukacs \cite[p. 192]{src:Lukacs:CharacteristicFunctions}, the moment-generating functions of the pushforward measures $\phi_* \mu$, $(\theta \phi)_* \mu$ and $(\phi - \theta \phi)_* \mu$ exist as integrals in some neighbourhood of zero.
Consequently, we can take $\lambda \to 0$ under the integral and obtain
\begin{equation}
\lim_{\lambda \to 0}
\sum_{m,n = 1}^{2} c_m^* \hat{\mu} \left( \psi_m - \theta \psi_n \right) c_n
=
\int_{\mathcal{D}^*_\beta} T \left( \phi \right) T \left( \theta \phi \right) \mathrm{d} \mu \left( T \right)
=
\left< \phi, \theta \phi \right>_{L^2(\mu)} \ge 0 \, .
\end{equation}
\end{proof}
\begin{theorem}[{\cite[Theorem 6.2.2]{src:GlimmJaffe}}]
Let $\mu$ be a $\theta$-invariant Gaußian measure on $\mathcal{D}^*_\beta$.
Then $\mu$ is reflection positive if and only if $( \phi, \theta \phi )_{L^2(\mu)} \ge 0$ for all $\phi \in \mathcal{D}_+$.
\end{theorem}
\begin{proof}
$\Rightarrow$: This is clear by the preceeding \namecref{thm:ReflectionPositiveL2}.\\
$\Leftarrow$: Let $(\cdot, \cdot)$ denote the inner product in $L^2(\mu)$ and let $(\phi_n)_{n \in \mathbb{N}}$ be a sequence in $\mathcal{D}_+$, $(c_n)_{n \in \mathbb{N}}$ a sequence of complex numbers and $N \in \mathbb{N}$.
Then, $\theta$-invariance implies
\begin{equation}
\sum_{m,n = 1}^{N} c_m^* \hat{\mu} \left( \psi_m - \theta \psi_n \right) c_n
=
\sum_{m,n = 1}^{N} c_m^* \hat{\mu} \left( \phi_m \right) \exp \left[ ( \phi_m, \theta \phi_n ) \right] \hat{\mu} \left( \phi_n \right) c_n \, .
\end{equation}
Since $\hat{\mu}$ is real, the statement follows if $( \exp \left[ ( \phi_m, \theta \phi_n ) \right] )_{m,n = 1}^N$ is a positive semidefinite matrix.
Since $( \phi_m, \theta \phi_n ) = ( \theta \phi_m, \phi_n )$ by the $\theta$-invariance of $\mu$, $\theta$ extends to a positive semidefinite linear operator on the complexification of $\mathrm{span} \{ \phi_n : n \in \mathbb{N} \}$.
Consequently, $(( \phi_m, \theta \phi_n ) )_{m,n = 1}^N$ is positive semidefinite.
By decomposing the exponential as a power series, the claim now follows from \cref{lem:PositiveSemidefinitenessPreserving}.
\end{proof}
The main theorem of this article depends on the following simple property of a function with respect to $\theta$.
\begin{definition}
\label{def:ThetaSplitting}
A function $F : \mathcal{D}^*_\beta \to [-\infty,\infty]$ is called \textbf{$\bm{\theta}$-splitting} if there exists a function $G : (\mathcal{D}_+)^*_\beta \to [-\infty,\infty]$ such that
\begin{equation}
F = G \circ \pi_+ + G \circ \pi_+ \circ \theta \, .
\end{equation}
\end{definition}
\begin{theorem}
\label{thm:ReflectionPositivity}
Let $\mu$ be a $\theta$-invariant reflection positive centred Gaußian measure on $\mathcal{D}^*_\beta$.
Then, for any $\mu$-measurable $\theta$-splitting function $F : \mathcal{D}^*_\beta \to [-\infty,\infty]$ with $exp \circ F \in L^1(\mu)$, the finite Borel measure
\begin{equation}
\omega = \exp \left[ F \right] \cdot \mu
\end{equation}
is reflection positive.
\end{theorem}
\begin{proof}
Define
\begin{equation}
j : \mathcal{D}^*_\beta \to ( \mathcal{D}_+ )^*_\beta \times ( \mathcal{D}_+ )^*_\beta \qquad
T \mapsto \left( \pi_+ T , \pi_+ \theta T \right) \, .
\end{equation}
$j$ is clearly continuous such that the pushforward measure $j_* \mu$ is a Radon measure $\nu$ on $( \mathcal{D}_+ )^*_\beta \times ( \mathcal{D}_+ )^*_\beta$.
Now, let
\begin{equation}
F_2 : ( \mathcal{D}_+ )^*_\beta \times ( \mathcal{D}_+ )^*_\beta \to \mathbb{R} \qquad
\left( T, K \right) \mapsto G \left( T \right) + G \left( K \right) \, .
\end{equation}
Then, for every $T \in \mathcal{D}^*_\beta$,
\begin{equation}
\left( F_2 \circ j \right) \left( T \right)
=
G \left( \pi_+ T \right) + G \left( \pi_+ \theta T \right)
=
F \left( T \right) \, ,
\end{equation}
such that $F_2$ is $\nu$-measurable by \cref{cor:PushforwardMeasurability}.
Turning to reflection positivity, let $(\phi_n)_{n \in \mathbb{N}}$ be a sequence in $\mathcal{D}_+$ and note that
\begin{equation}
\begin{aligned}
\hat{\omega} \left( \phi_m - \theta \phi_n \right)
&=
\int_{\mathcal{D}^*_\beta}
\exp \left[ i T \left(\phi_m \right) - i T\left( \theta \phi_n \right) + F \left(T \right) \right] \mathrm{d} \mu \left(T \right) \\
&=
\int_{\mathcal{D}^*_\beta}
\exp \left[ i \; j \left( T \right) \left( \phi_m, -\phi_n \right) + \left( F_2 \circ j \right) \left(T \right) \right] \mathrm{d} \mu \left(T \right) \\
&=
\int_{(( \mathcal{D}_+ )^*_\beta)^2}
\exp \left[ i T \left( \phi_m \right) - i K \left( \phi_n \right) + F_2 \left(T, K \right) \right] \mathrm{d} \nu \left(T, K \right) \\
&=
\int_{(( \mathcal{D}_+ )^*_\beta)^2}
\exp \left[ i T \left( \phi_m \right) - i K \left( \phi_n \right) + G \left(T \right) + G \left( K \right) \right] \mathrm{d} \nu \left(T, K \right) \, .
\end{aligned}
\end{equation}
The above expression suggests to find a disintegration of $\nu$ that separates the $T$ and $K$ variables.
To that end, recall that $\mu$ is Gaußian such that for any $\phi, \psi \in \mathcal{D}_+$, we have
\begin{equation}
\hat{\nu} \left( \phi, \psi \right)
=
\int_{\mathcal{D}^*_\beta} \exp \left[ i T \left( \phi \right) + i T \left( \theta \psi \right) \right] \mathrm{d} \mu \left( T \right)
=
\exp \left[ - \frac{1}{2} \left\Vert \phi + \theta \psi \right\Vert_{L^2(\mu)}^2 \right] \, .
\end{equation}
Furthermore, by \cref{thm:ReflectionPositiveL2}, Cauchy-Schwartz and the $\theta$-invariance of $\mu$,
\begin{equation}
0 \le \left< \phi, \theta \phi \right>_{L^2(\mu)} \le \left< \phi, \phi \right>_{L^2(\mu)} \, .
\end{equation}
Moreover, since $(\mathcal{D}_+)^*_\beta$ is a reflexive, nuclear, barrelled space, there exist uniquely determined Radon Gaussian measures $P$ and $Q$ on $(\mathcal{D}_+)^*_\beta$ with
\begin{align}
\hat{P} \left( \phi \right) &= \exp \left[ - \frac{1}{2} \left< \phi, \phi \right>_{L^2(\mu)} + \frac{1}{2} \left< \phi, \theta \phi \right>_{L^2(\mu)} \right] \, , \\
\hat{Q} \left( \phi \right) &= \exp \left[ - \frac{1}{2} \left< \phi, \theta \phi \right>_{L^2(\mu)} \right]
\end{align}
by Minlos theorem \cite[Theorem 7.13.9]{src:Bogachev:MeasureTheory2}.
Defining the diagonal map
\begin{equation}
\Delta : (\mathcal{D}_+)^*_\beta \to (\mathcal{D}_+)^*_\beta \times (\mathcal{D}_+)^*_\beta \qquad T \mapsto \left(T, T \right)
\end{equation}
it is clear that
\begin{equation}
\hat{\nu} \left( \phi, \psi \right)
=
\hat{P} \left( \phi \right) \hat{P} \left( \psi \right)
\hat{Q} \left( \phi + \psi \right)
=
\hat{P} \left( \phi \right) \hat{P} \left( \psi \right)
\widehat{\Delta_* Q} \left( \phi, \psi \right)
\end{equation}
for all $\phi, \psi \in \mathcal{D}_+$.
Equivalently, $\nu = (P\times P) * (\Delta_* Q)$ by \cref{thm:SouslinBorelProduct}.
Hence, it is straightforward to verify that
\begin{equation}
\begin{aligned}
\hat{\omega} \left( \phi_m - \theta \phi_n \right)
=
\int_{(( \mathcal{D}_+ )^*_\beta)^3}
\exp &\big[
i \left( T + L \right) \left( \phi_m \right) - i \left( K + L \right) \left( \phi_n \right) \\
&+ G \left(T + L \right) + G \left( K + L \right)
\big] \mathrm{d} \left( P \times P \times Q \right) \left(T, K, L \right) \, .
\end{aligned}
\end{equation}
Now, the functions
\begin{equation}
H_m \left(L \right)
=
\int_{( \mathcal{D}_+ )^*_\beta}
\exp \left[ - i \left( T + L \right) \left( \phi_m \right) + G \left( T + L \right) \right] \mathrm{d} P \left(T \right)
\end{equation}
for $m \in \mathbb{N}$ are well-defined $Q$-almost everywhere.
Thus, using Fubini, we arrive at
\begin{equation}
\sum_{m,n = 1}^N c_m^* \hat{\omega} \left( \phi_m - \theta \phi_n \right) c_n
=
\int_{( \mathcal{D}_+ )^*_\beta} \left\vert \sum_{n=1}^N c_n H_n \left( L \right) \right\vert^2 \mathrm{d} Q \left(L \right)
\ge 0
\end{equation}
for any $N \in \mathbb{N}$ and any sequence $(c_n)_{n \in \mathbb{N}}$ of complex numbers.
\end{proof}
This theorem is strikingly simple and can be applied very easily.
Let us call a locally convex space $X$ together with a continuous, linear map $j : X \to \mathcal{D}^*_\beta$ a \textbf{\bm{$\theta$}-model space}, if there is a continuous, linear operator (slight abuse of terminology) $\theta : X \to X$ such that $\theta \circ j = j \circ \theta$.
\begin{example}
Examples of such $\theta$-model spaces are e.g. function spaces on $\theta$-symmetric lattice subsets of $\mathbb{R}^{d+1}$, $\mathcal{D}$ or the space of Schwartz functions on $\mathbb{R}^{d+1}$ together with their respective usual injections into $\mathcal{D}^*_\beta$.
\end{example}
\begin{remark}
The above examples cover most of what is used in literature on Euclidean interacting quantum field theories and are also Souslin spaces.
\end{remark}
We may now extend the definition of a \textbf{$\bm{\theta}$-splitting} function to $\theta$-model spaces.
\begin{definition}
\label{def:ThetaSplittingGeneral}
A function $F : X \to [-\infty,\infty]$ on a $\theta$-model space $(X,j)$ is called \textbf{$\bm{\theta}$-splitting} if there exists a function $G : X \to [-\infty,\infty]$ such that
\begin{equation}
F = G \circ \pi^X_+ + G \circ \pi^X_+ \circ \theta
\end{equation}
Here, $\pi^X_+ : X \to X / j^{-1}(\ker \pi_+)$ is the canonical quotient map.
\end{definition}
\begin{corollary}
Let $(X, j)$ be a Souslin $\theta$-model space.
Furthermore, let $\mu$ be a Gaußian measure on $X$ with the property that $j_* \mu$ is $\theta$-invariant and reflection positive.
Then, for any $\mu$-measurable $\theta$-splitting function $F : X \to [-\infty,\infty]$ with $exp \circ F \in L^1(\mu)$, the finite Borel measure
\begin{equation}
\omega = j_* \left( \exp \left[ F \right] \cdot \mu \right)
\end{equation}
is reflection positive.
\end{corollary}
\begin{proof}
Let $G$ and $\pi^X_+$ be given as in \cref{def:ThetaSplittingGeneral} and define the function $G_2 : (\mathcal{D}_+)^*_\beta \to [-\infty,\infty]$ given by
\begin{equation}
T \mapsto \begin{cases}
G( \pi^X_+ x ) & \text{if } \exists \, x \in X : T = \pi_+ j x \\
0 & \text{else.}
\end{cases}
\end{equation}
To see that $G_2$ is well-defined, note that if $\pi_+ j x = \pi_+ j y$ for some $x,y \in X$, we have that there is some $T \in \ker \pi_+$ with $j (x - y) = T$, i.e. $x - y \in j^{-1}(\ker \pi_+) = \mathrm{ker} \, \pi^X_+$.
Now, define the function $F_2 : \mathcal{D}^*_\beta \to [-\infty,\infty]$ given by
\begin{equation}
T \mapsto G_2 \left( \pi_+ T \right) + G_2 \left( \pi_+ \theta T \right) \, .
\end{equation}
Clearly, $F_2 \circ j = F$ such that $F_2$ is $(j_* \mu)$-measurable by \cref{cor:PushforwardMeasurability}.
Consequently, $\omega = \exp[F_2] \cdot (j_* \mu)$ and \cref{thm:ReflectionPositivity} applies.
\end{proof}
We finish this article by a simple example.
\begin{example}
Let $\mathcal{S}$ denote the space of Schwartz functions on $\mathbb{R}^{d+1}$.
Define $j : \mathcal{S} \to \mathcal{D}^*_\beta$ by $j(\phi)(\psi) = \int_{\mathbb{R}^{d+1}} \psi \phi$ for all $\phi \in \mathcal{S}$ and $\psi \in \mathcal{D}$.
Moreover, let $\mu$ be a Gaußian measure on $\mathcal{S}$ with the property that $j_* \mu$ is $\theta$-invariant and reflection positive.
Note that this excludes the Gaußian measure on the space $\mathcal{S}^*$ of tempered distributions modelling the Klein-Gordon field.
However, regularised versions of that measure will work, see e.g. \cite[Example 6.2]{src:Ziebell:RigorousFRG}.
Furthermore, let $F: \mathcal{S} \to \mathbb{R}, \phi \mapsto -\lambda \int_{\mathbb{R}^{d+1}} \phi^4$ for some $\lambda > 0$.
Then,
\begin{equation}
F(\phi)
=
-\lambda \int_{\mathbb{R}_{> 0} \times \mathbb{R}^d} \phi^4
-\lambda
\int_{\mathbb{R}_{> 0} \times \mathbb{R}^d} \left( \theta \phi \right)^4
\end{equation}
provides a $\theta$-splitting of $F$.
\end{example}
\section*{Acknowledegments}
This work has been supported by the Deutsche Forschungsgemeinschaft
(DFG) under Grant No. 406116891 within the Research Training Group RTG 2522/1.
\appendix
\printbibliography
\end{document}